\documentclass[12pt,letterpaper]{article}
\usepackage[centertags]{amsmath}
\usepackage{amsfonts,amsthm,amssymb}
\usepackage{amssymb}
\usepackage{amsmath}
\usepackage{graphicx}
\usepackage{booktabs}
\usepackage{appendix}
\usepackage[hmargin=2.6cm,vmargin=2.6cm]{geometry}
\usepackage{etoolbox}
\usepackage{tikz}
\usepackage{float}
\usepackage{natbib}
\usepackage[todonotes={textsize=tiny}]{changes}
\usepackage{mathtools}
\usepackage[unicode=true,
 bookmarks=false,
 breaklinks=false,pdfborder={0 0 0},pdfborderstyle={},backref=false,colorlinks=true]
 {hyperref}
\hypersetup{
 citecolor=blue, linkcolor=blue, urlcolor=blue}
 
 \usepackage[T1]{fontenc}
\usepackage[latin9]{inputenc}
\usepackage{geometry}
\usepackage{color}
\usepackage{array}
\usepackage{amssymb}
\usepackage{setspace}

\usepackage{xcolor,framed}
\theoremstyle{definition}
\def\supp{{\textrm{supp}}}

\newcommand{\p}[1]{\mathrm{P}\left[#1\right]}

\newcommand{\abs}[1]{\left| #1 \right|}

\newcommand{\bias}{\mathrm{bias}}

\newcommand{\fa}{\mathfrak{a}}
\newcommand{\fb}{\mathfrak{b}}

\newtheorem{lemma}{Lemma}

\newtheorem{proposition}{Proposition}

\newtheorem{theorem}{Theorem}

\newtheorem{claim}[theorem]{Claim}

\newtheorem{corollary}{Corollary}

\newtheorem{remark}{Remark}[section]

\newtheorem*{lem:main}{Lemma \ref{lem:main}}
\newtheorem{fact}{Fact}
\newcommand{\eps}{\varepsilon}

\begin{document}
\title{Social Media and Democracy
}
\author{Ronen Gradwohl\thanks{Department of Economics and Business Administration, Ariel University.  \protect\href{mailto:roneng@ariel.ac.il}{roneng@ariel.ac.il}. }  \and Yuval Heller\thanks{Department of Economics, Bar-Ilan University \& University of California San Diego. 
\protect\href{mailto:yuval.heller@biu.ac.il}{yuval.heller@biu.ac.il}. } \and Arye Hillman\thanks{Department of Economics, Bar-Ilan University. \protect\href{mailto:
arye.hillman@biu.ac.il}{arye.hillman@biu.ac.il}. }}
\date{\today}

\maketitle
\begin{abstract}
We study the ability of a social media platform with a political agenda to influence voting outcomes. Our benchmark is Condorcet's jury theorem, which states that the likelihood of a correct decision under majority voting increases with the number of voters. We show how information manipulation by a social media platform can overturn the jury theorem, thereby undermining democracy. We also show that sometimes the platform can do so only by providing information that is biased in the {\em opposite direction} of its preferred outcome. Finally, we compare  manipulation of voting outcomes through social media to manipulation through traditional media.
\end{abstract}


\pagebreak
\section{Introduction}\label{sec:introduction}
In early theories of voting \citep{hotelling1929stability, downs1957economic}, voters are fully informed and vote based on their given preferences. In subsequent expositions, however, voters do not know which policies or competing political candidates merit their support. They are then susceptible to influence by political advertising \citep{hillman1988domestic,grossman2001special,prat2002campaign}
and by competing endorsements of policies and candidates  \citep{grossman1999competing,lichter2017theories} through mass media. 

Voters can also be influenced by social media. Rather than competition between media institutions, network externalities create natural monopolies for users of social media. In addition, rather than being exposed to the same, public information, users of social media face private persuasion through algorithms that select and order the information available to them. Facebook's News Feed algorithm, for example, determines personally, for each user, the posts  are visible and the order in which they appear \citep{facebook}. Users of social media can also be influenced by censorship or by delays in access to information. Twitter, for example, places warning labels on tweets and sometimes blocks information dissemination by locking accounts \citep{twitter1, twitter2}.

In this paper, we study the ability of a monopoly social media platform with a political agenda to influence users' views through such selective transmission of information. In our model voters may have outside information and preferences that are independent of the social media platform's agenda; however, voters are not perfectly informed, and so do not know with certainty which candidate or policy to support. To what extent can the social media platform exploit this incomplete information?
	
Our beginning is Condorcet's jury theorem (\citeyear{de2014essai}), which shows the merit of 
majority voting. In the setting of the theorem, voters face a choice between two alternatives, A and B. One alternative is objectively better for all voters, but voters do not know a priori which is the preferred alternative. Each voter has limited information, expressed in probabilities, indicating which alternative is better. The jury theorem states that the better alternative is more likely to be chosen when the decision is made by majority voting than when a single individual makes the decision. As the number of voters increases, the likelihood of the decision under majority voting being correct approaches certainty. The theorem therefore vindicates democracy. 

We ask, in an extension of the theorem: Does a group of voters who obtain information through a social media platform that has a political agenda also necessarily successfully make the correct decision by majority voting? Or, can information manipulation by social media undermine the merit of majority voting?

We undertake our study and derive our results within the Bayesian persuasion framework \citep[see][for a survey]{kamenica2019bayesian}. The media platform commits to how it will manipulate information, a commitment that in practice is implemented through the algorithms used by media platforms to select and order the information that individuals obtain. Individuals are aware of the manipulation and make correct inferences using Bayes' rule based on the information available. We show that information manipulation can lead majority voting to result in the choice of the alternative sought by the social media platform rather than the alternative that is preferable for voters. Contrary to Condorcet's jury theorem,
the consequence of information manipulation can be that a single individual is more likely to
choose the socially preferred alternative than a group that decides by majority voting, even
when individuals are aware of and account for the platform's information manipulation. Democracy is thus undermined.

In the remaining parts of the introduction we demonstrate our main results using examples, and then briefly describe our general model and results. Section \ref{sec:lit} discusses the related literature. In Section \ref{sec:model} we formally present our model, and then, in Section \ref{sec:results}, we describe our results. Various extensions are presented in Section \ref{sec:extensions}, followed by a conclusion in Section \ref{sec:conclusions}.

\subsection{Illustrative Example}\label{sec:illustrative-example}
Before we proceed to the formal model we provide an illustrative example. We begin with Condorcet's jury theorem, proceed to the basic logic of Bayesian persuasion, and then show how social-media persuasion can overturn the conclusion of the theorem.

\paragraph{Condorcet's Jury Theorem}
There are two alternatives, A and B, one of which is ``better'' or socially preferred. Voters are a priori unsure about which of A and B is better, but they do have some limited prior information supporting one of the alternatives. Suppose that, independently for each voter, this prior information has accuracy 55\%---namely, there is a 55\% chance of it being correct. 
Specifically, if a voter has information indicating that one of the alternatives (say, A) is better, then this voter believes there is a 55\% chance that A is, in fact, better (and a 45\% chance that B is better). 

Thus, if 
each voter chooses the alternative that is more-likely to be better based on this prior information, then she has a 55\% chance of choosing the socially preferred alternative. 
Rather than a single voter deciding between A and B, suppose that the decision is made by majority voting. Condorcet's jury theorem states that, with sufficiently many voters, if each voter votes for the alternative that is more-likely to be better based on her own information, then the correct alternative is chosen with near certainty (due to the Law of Large Numbers).

\paragraph{Bayesian Persuasion}
We now introduce an information designer (the social media platform) that can send voters additional information in the form of a state-dependent message. The designer is interested in ensuring that alternative A is chosen,
regardless of which alternative is socially better. Suppose the designer can send one of two messages---either  ``A is better'' or ``B is better''. 
The designer has to choose a rule that determines which message to send as a function of the true socially-better alternative.
Importantly, we assume ($\grave{\textrm{a}}$ la \citealp{kamenica2011bayesian})
that the voters know the designer's rule. The simple rule of always sending ``A is better'' is useless because  the voters are aware that this signal is non-informative. A designer who wants to influence voters to vote for A would instead have to be more subtle.

Consider the following messaging scheme:
if A is better, the designer always sends the message ``A is better''; if B is better, the designer sends the message ``B is better'' with probability 30\%, and the ``incorrect'' message ``A is better'' with the remaining 70\% probability. What can voters infer from the designer's message? Observe first that if a voter obtains message ``B is better'' from the designer, then this voter can be certain that B is, in fact, better. 

By contrast, if a voter obtains message ``A is better'', then the voter cannot be certain about which alternative is better, since this message is sent both when A is actually better but also sometimes when B is better. The exact inference that a voter can make upon obtaining message ``A is better'' depends on this voter's prior information---whether it supports A or B---and is calculated by Bayes' rule. 
A simple calculation shows that a voter whose prior information supported A (i.e., she has a prior belief that $A$ is better with probability $55\%$), updates her posterior belief that $A$ is better to $64\%$:
\begin{align*}
\p{A|\mbox{``A is better''}} &=  \frac{\p{A\cap\mbox{``A is better''}}}{\p{\mbox{``A is better''}}}=\frac{100\%\times55\%}{100\%\times55\%+70\%\times45\%}\approx 64\%.
\end{align*}
Similarly, a voter whose prior  supported B updates her posterior belief that  $A$ is better to $54\%$:
\begin{align*}
\p{A|\mbox{``A is better''}} &=  \frac{\p{A\cap\mbox{``A is better''}}}{\p{\mbox{``A is better''}}}=\frac{100\%\times45\%}{100\%\times45\%+70\%\times55\%}\approx 54\%.
\end{align*}

\paragraph{Persuasion Overturns Condorcet}
We now consider the effect of the designer's message on voting outcomes, and show that the consequence is to overturn the necessity of the conclusion of Condorcet's jury theorem.
	Consider first a single voter who makes the decision based on which alternative is more likely to be better. Recall from the previous section that if this voter obtained message ``B is better'' from the designer, 
	then the voter is certain that $B$ is better, and will choose alternative B. If the voter obtained message ``A is better'', however, then her belief about the likelihood of $A$ being  better is either 64\% or 54\%, depending on her initial information. Importantly, however, {\em in both cases} A is more likely to be better, and so in both cases the voter will choose A upon obtaining this message.
	
	How likely is this voter to choose the alternative that is actually better? Recall that if A is actually better, then the designer always sends message ``A is better'', and the voter makes the correct choice and votes for $A$. However, if B is actually better, then the designer sends message ``B is better'' only with probability 30\%. If this happens, then the voter is certain that B is better and makes the correct choice, B. However, with the remaining 70\% probability, the designer sends message ``A is better'', after which the voter chooses the incorrect alternative, A. Thus, if B is actually better, the voter makes the correct choice with probability 30\% and the incorrect choice with probability 70\%.
	
Rather than a single voter deciding between A and B, suppose that the decision is made by majority voting 
in a large population. If A is actually better then, as above, every voter obtains message ``A is better'' and votes for A. 
If B is actually better, however, then on average 30\% of voters obtain message ``B is better'' and vote for B, whereas on average 70\% of voters obtain message ``A is better'' and vote for A. In this case, by the Law of Large Numbers, if there are sufficiently many voters then with very high probability a majority votes for A, although A is the {\em incorrect} alternative. Thus, if B is better, majority voting never leads to the correct alternative being chosen. Note that this is worse than a single voter making the decision, as then there was a 30\% chance that B is chosen when it is better.
	
	In this example we have illustrated how additional information provided by a self-interested designer can overturn the conclusion of Condorcet's jury theorem. When there is no designer, a large group deciding by majority rule is more likely to make the correct decision than a single voter deciding alone. With additional information from a designer, however, the conclusion is reversed: 
	A single voter makes the correct decision with  probability 30\% if B is better, but majority voting never leads to the correct decision with when B is better (and in both cases the correct decision is made when $A$ is better).

\paragraph{Negatively Biased Signals and the Limits of Persuasion}
We now expand on the example above in order to answer two questions. First, what are the properties of the messaging scheme used by the designer to overturn an election? In particular, is it necessarily biased towards the outcome desired by the designer? And second, can the designer always overturn the election, or are there  limits to persuasion?

We begin with the first question. Observe that in the messaging scheme described above the designer sends the message ``A is better'' more often than the message ``B is better''.
This scheme is thus {\em biased} towards A. In the example above, this particular scheme is not unique, and there are other messaging schemes---some of which are biased towards B rather than A---that also allow the designer to overturn the election. We now describe such a scheme.

Suppose that, if A is better, the designer sends the message ``A is better'' with probability 53\%, and the message ``B is better'' with probability 47\%. If B is better, the designer sends the message ``B is better'' with probability 57\% and the message ``A is better'' with probability 43\%. Notice that the designer is more likely to lie in favor of B---to send message ``B is better'' when A is actually better (47\%)---than to lie in favor of A---to send message ``A is better'' when B is actually better (43\%). This messaging scheme is thus biased towards B.

What is the effect of this scheme? Consider first a voter whose prior information supported A, so that this voter initially believes A to be better with probability 55\%. If this voter now obtains message ``B is better'' from the designer, the voter updates the probability of A being correct by conditioning on the obtained message. The new belief is:

\begin{align*}
\p{A|\mbox{``B is better''}} &=  \frac{\p{A\cap\mbox{``B is better''}}}{\p{\mbox{``B is better''}}}=\frac{47\%\times55\%}{47\%\times55\%+57\%\times45\%}\approx 50.2\%.
\end{align*}

Thus, a voter who initially believed A to be better with probability 55\%, updates this belief to 50.2\% upon obtaining message ``B is better'' from the designer. Critically, this updated belief is above 50\%, and so the voter still believes that A is more likely to be the better alternative. Similarly, if this same voter obtained message ``A is better'' instead of ``B is better'', the updated belief would be even higher (around 60\%).

Consider now a voter whose prior information supported B, so that this voter initially believes B to be better with probability 55\% (and A to be better with probability 45\%). If this voter now obtains message ``A is better'' from the designer, the voter updates the probability of A being correct by conditioning on the obtained message. The new belief is:
\begin{align*}
\p{A|\mbox{``A is better''}} &=  \frac{\p{A\cap\mbox{``A is better''}}}{\p{\mbox{``A is better''}}}=\frac{53\%\times45\%}{53\%\times45\%+43\%\times55\%}\approx 50.2\%.
\end{align*}

Thus, a voter who initially believed A to be better with probability 45\%, updates this belief to 50.2\% upon obtaining message ``A is better'' from the designer. Critically, this updated belief is also above 50\%, and so the voter now believes A is more likely to be the better alternative. 

In fact, the only voters who vote for B are those whose prior information supported B, and who then additionally obtained message ``B is better'' from the designer. All other voters end up with a belief indicating that A is more likely to be better, and who then vote for A.

The final step is to determine the fraction of voters for A. If A is the better alternative, then only 45\% of voters' initial information supports B, and of those, only 47\% obtain message ``B is better''. Thus, on average only $45\%\times 47\%=21\%$ of voters vote for B, and the rest for A. With high probability, then, the majority vote for A. 
Similarly,  if B is the better alternative, then 55\% of voters' initial information supports B, and of those,  57\% obtain message ``B is better''. Thus, on average  $55\%\times 57\%=31\%$ of voters vote for B. Again, this implies that with high probability, the majority vote for A. Thus, regardless of which of A or B is actually better, the designer's messaging scheme leads to alternative A always being elected.

The example just described shows that the designer can overturn the election with a messaging scheme biased towards B---that is, towards the {\em opposite} alternative of the one preferred by the designer. In the previous section we showed that the election can also be overturned with a scheme biased towards A. So which will the designer use?

In Section \ref{sec:results} we characterize the biases of the schemes that can be used to overturn the election, and how they are related to the accuracy of voters' prior information. In particular, in Corollary~\ref{cor-homogenous} we show that, if this accuracy is between $1/2$ and $2/3$, then the designer can use both classes of schemes 
to overturn the election.  
However, we also show that if the accuracy is between $2/3$ and $\sqrt{2}/2$, then {\em only} schemes that are biased towards B can overturn the election. That is, in this case, the only way for the designer to overturn the election is to design a messaging scheme that is more likely to give incorrect information supporting B.

Finally, Corollary~\ref{cor-homogenous} also points to the limits of persuasion. The corollary states that if the voters' prior information has accuracy {\em greater than} $\sqrt{2}/2$ then there is no messaging scheme that allows the designer to overturn the election. The intuition is the following. In order to overturn the election, the designer needs to 
convince enough voters whose initial information supported B to switch to A. 
But the more accurate the voters' prior information, the more difficult it is to convince them to switch.
At the extreme, if voters' prior information is {\em perfectly} accurate, and so they know with certainty which of the alternatives is actually better, then clearly no additional information from the designer could ever shift any of their beliefs. In 
Corollary~\ref{cor-homogenous} we show that even with less extreme accuracies, if the accuracies are above $\sqrt{2}/2$ then there is no messaging scheme that allows the designer to convince a majority 
of the voters to vote for A.

\subsection{A General Model}
Our general model, of which the above example is a special case, shows how in more encompassing circumstances the designer is able to sway the vote towards his desired outcome. As in the example, in some cases, the preferred outcome of the designer is chosen with probability one. In other cases, the designer is unable to exert any influence at all. The ability of the designer to influence outcomes depends both on the extent to which users are exogenously informed and on the heterogeneity in the accuracy levels of users' exogenous information. 

We begin with a setting in which some voters are uninformed, whereas others are exogenously informed with the same accuracy. We show that, as long as this accuracy is not too high, the social media platform can completely sway the outcome of the vote. Interestingly, we also show that, in this case, the designer can sway the outcome of the vote with an unbiased signal. We also show that, in some cases, the designer can {\em only} sway the outcome with a signal biased towards B---that is, one that sends message ``B is better'' more often than the message ``A is better''. Finally, we show that when users' exogenous signal accuracy is high, the designer can no longer affect the outcome.

Next, we extend the model to a scenario in which there are two kinds of voters corresponding to two levels of exogenous-signal accuracy. Here we show that if both accuracies are sufficiently low, the designer can determine the outcome of the vote with probability one; if both accuracies are sufficiently high, the designer cannot exert any influence; and if accuracies are in an intermediate range, then whether or not the designer has influence depends on both the specific accuracies of information and the fractions of users of each kind. We show that, in this last case, there is a nonmonotonicity wherein the designer can determine the outcome if most users are of one kind and if most users are of the other kind, but not when there is a substantial fraction of each kind. There is also
a nonmonotonicity wherein the designer can determine the outcome of the vote when the two accuracies are sufficiently close or sufficiently far, but not in between. Finally, we provide conditions under which the designer can control the outcome of the vote only with a signal that is biased towards A and only with a signal that is biased towards B.

\section{Related Literature}\label{sec:lit}
A literature on social media and transmission of information to voters through the lens of Bayesian persuasion and information design has focused on studying the extent to which a holder of information can manipulate others by selective disclosure of  information
\citep{kamenica2011bayesian}. In the terminology of this literature, we study a setting of private Bayesian persuasion with multiple, informed receivers. We also restrict our designer to conditionally independent signals. 

A result that Condorcet's jury theorem can be overturned has been derived based on costly acquisition of information  \citep{mukhopadhaya2003jury,koriyama2009resurrection}. We retain the assumption of the theorem that individuals incur no cost in acquiring information. On the contrary, we describe information as imposed upon individuals through the social media platform without individuals actually seeking information.\footnote{For scenarios with costly voting, see \citet{sun2019theory}, who study the effects of public persuasion and cheap talk on voter turnout and \citet{chan2019pivotal}, 
who study the impact of costly voting on uninformed voters.}

A substantial literature conjoins an information designer with strategic voting, with conclusions relying critically on the interplay between the two (see next paragraph for details). In contrast, we view voters as personally optimally voting for the outcome that they believe is better or as voting ``sincerely''.   This allows us to focus on and cleanly characterize the direct effects of persuasion on voting outcomes. Strategic voting, in contrast, requires the assumption that voters believe that their vote can be pivotal (or decisive) in determining the outcome of a voting.\footnote{See, for example, \citet{pons2018expressive} and \citet{spenkuch2018expressive} and the references therein for the empirical literature on strategic voting. Although much of the literature detects some amount of strategic voting, in most instances the large majority of voters are not strategic.} A concession to realism is that voters in usual two-party elections for political representatives know that the likelihood of their one vote being decisive is essentially zero. We assume that voters vote sincerely based on their information and do not regard themselves as possibly decisive or pivotal.\footnote{A paradox of voting is present \citep{downs1957economic} 
if voters have a negligible probability of being pivotal; what reason does the voter have to vote? A benefit imputed to voting is expressive utility \citep{hillman2010expressive}. The act of voting is used to express an identity or sense of belonging. In this context, voting is rational.}   We suppose that voters enjoy an expressive benefit to voting for the correct alternative, and that all voters vote with no expectation that their vote will be pivotal in determining an election outcome.

 \citet{austen1996information} study strategic voting in the context of Condorcet's jury theorem; they describe strategic voters conditioning their behavior on being pivotal in an election, and show that it may no longer be personally optimal to vote sincerely. \citet{feddersen1998convicting} show that under general voting rules, such strategic voting may lead to a failure of information aggregation, but \citet{feddersen1997voting} show that, nonetheless, the jury theorem holds:  the probability that the outcome of the vote under majority voting is correct approaches certainty as the number of voters increases. \citet{heese2021persuasion} start with  the general strategic voting framework of \citet{feddersen1997voting}, and add an information designer. They show that, under strategic voting,  the designer can provide additional information in a way that leads to any preferred outcome with high probability. \citet{taneva2019information} studies a symmetric two-voter game; her main result is that, in this setting, it is never optimal for the designer to commit to a conditionally independent signal. More generally,  \citet{arieli2019private}  study a general framework for private persuasion, and characterize optimal solutions under submodularity assumptions on designer's utility. In our setting the utility function is neither supermodular nor submodular, so these results do not apply. Finally, \citet{kerman2020persuading} focus on persuading uninformed voters through correlated signals that induce sincere rather than strategic voting in equilibrium.

Other literature has studied public persuasion, where the designer commits to a signal and all voters observe the same realization, as opposed to private persuasion through social media where each voter obtains a personally distinct information realization.  Public Bayesian persuasion is an appropriate model for traditional media such as newspapers and television. Private persuasion is appropriate for social media, which can use personal information garnered from individuals' attributes and communications to design selective information directed at an individual. The literature includes 
 \citet{alonso2016persuading}, who study public persuasion of heterogeneous voters.
\citet{wang2013bayesian} compares public signals to private, conditionally independent signals, and shows that the former is more informative (and hence worse for the designer) than the latter.  
 \citet{gitmez2022polarization} also study public persuasion, with a focus on the interplay between the bias of the signal and voters' polarization. While less related to our work in terms of focus and results, the model of  \citet{gitmez2022polarization} is similar to ours in that there is a continuum of voters and each votes sincerely.

A sizable literature has studied bias in traditional and social media, both theoretically and empirically \citep[see][for a variety of surveys]{gentzkow2015media,puglisi2015empirical,stromberg2015media}. Much of this literature focuses on the informational effects of media bias, as well as its welfare implications. To the best of our knowledge, our insight that the 
manipulator's optimal bias 
is in the opposite direction of the intended outcome is new.

Last, in other related literature, \citet{denter2021social} study media bias when those it targets also have exogenous information, and share this information via a social network. They study the effects of the network's connectivity when voters suffer from correlation neglect on the optimal bias of media. \citet{sun2019theory} study the effects of public persuasion and cheap talk on turnout in a model with costly voting. Relatedly, \citet{levy2021social} perform a field experiment to study the effects of social media (specifically, Facebook) on user's political leanings. 

\section{Model}\label{sec:model}
There are two policies (or candidates) denoted $A$ and $B$, and a continuum of voters who must decide between them.
There are two, equally likely
states of the world, $\Theta=\{\theta_A,\theta_B\}$, such that the better policy is $A$ in state $\theta_A$ and $B$ in 
state $\theta_B$. 
Each voter receives a symmetric binary signal about the state of the world. Signal realizations are either $a$ or $b$, and a signal of accuracy $q$ is one where 
$$\p{a|\theta_A}=\p{b|\theta_B} = q.$$ 
We allow heterogeneity in the signal accuracies. Specifically, we assume that a $\lambda$-fraction 
(where $\lambda\in(0,1)$) of the population has low signal accuracy $q_\ell\in[0.5,1)$, and the rest have high signal accuracy $q_h\in[q_\ell, 1]$. We say that the population is \emph{homogeneous} if $q_\ell=q_h$, and that it is \emph{heterogeneous} otherwise.\footnote{In Section \ref{sec:extensions} we allow more general distributions over signal accuracies and extend some of our results to this setup.}

In addition to the voters, there is an information designer.
The designer can send the voters additional informative signals that depend on the state of the world. Unlike the voters, however, the designer wants to maximize the probability that policy $A$ is chosen, regardless of the state. In particular, the designer's utility is 1 whenever $A$ is chosen, and 0 otherwise. In order to influence the voters, the designer chooses a binary signal $s$---a probabilistic function from the state to a pair $\{\mathfrak{a},\mathfrak{b}\}$---and sends each voter a conditionally independent realization of $s$. 
We interpret this signal as having been sent through a social media platform that is controlled by the information designer.
As in the Bayesian persuasion framework \citep{kamenica2011bayesian}, we assume that the designer commits to $s$ prior to learning the true state of the world.  In addition, and without loss of generality, assume that the signal $s$ is such that $\p{\theta_A|\fa}\geq \p{\theta_A|\fb}$. If this were not the case, one could simply switch the meaning of the two signal realizations. 

\begin{remark}
The assumption that the designer's signal is binary is without loss of generality, given that the set of actions per voter is binary. This follows from familiar ideas from information design, but is stated and proved formally in Lemma~\ref{lem:binary-signal-suffices} of Appendix~\ref{apx:binary-opt}. 
\end{remark}

Define the {\em bias} of the designer's signal as the weighted difference in the sizes of the ``lies'' in the two directions, namely,
$$\bias(s)=2\cdot(\p{s(\theta_B)=\mathfrak{a}}\cdot\p{\theta_B}-\p{s(\theta_A)=\mathfrak{b}}\cdot\p{\theta_A})=\p{s(\theta_B)=\mathfrak{a}}-\p{s(\theta_A)=\mathfrak{b}},$$
where the expression after the first equality presents the definition for a general prior, and the expression after the latter equality relies on having a uniform prior. Thus, under the uniform prior, symmetric signals in which $\p{\mathfrak{a}|\theta_A}=\p{\mathfrak{b}|\theta_B}$ are {\em unbiased}---that is, they have bias 0. 
Positive biases correspond to signals that are more likely to yield the signal $\mathfrak{a}$, and negative biases to those that are more likely to yield the signal $\mathfrak{b}$. 
Observe that for the canonical example of Bayesian persuasion, in which a prosecutor seeks to convince a judge that a defendant is guilty (which corresponds to state $\theta_A$), the optimal signal is one with positive bias (\citealp{kamenica2019bayesian}). We note that this definition of bias is an instance of the definition of \citet{gentzkow2015media} for media bias, under the assumption that the fully informative signal in which $\p{s(\theta_A)=\fa}=\p{s(\theta_B)=\fb}=1$ has bias 0.

A strategy $\sigma_i$ of voter $i$ is a function from signals (both the original signal and the 
designer's signal) to a distribution over two actions---a vote for $A$ or a vote for $B$. The outcome of the vote is then determined by the majority, with ties decided in favor of $A$.\footnote{Tie-breaking in favor of $A$ is assumed for simplicity only---the results are nearly identical for other tie-breaking rules; see Remark \ref{rem-tie-breaking}.} We assume that voting is {\em sincere}: each voter votes according to their own information, which includes the prior and the two signals. 
Sincere voting can either directly reflect the voters' preferences (namely, each voter obtains utility 1 if she votes for the better policy), or it can correspond to voters who care for the selected outcome (i.e., each voter obtains utility 1 if the better policy is chosen), yet when voting they do not take into account the strategic effects of pivotality.\footnote{The probability of being pivotal is decreasing in the population size (for finite populations), and it is equal to zero in our continuum-population model.}

Thus, each voter votes for $A$ if her posterior on state $\theta_A$ is above $1/2$ and votes for $B$ if her posterior is below $1/2$. As is standard in the Bayesian persuasion literature, we assume voters vote for $A$, the outcome preferred by the designer, when their posterior is exactly $1/2$; however, the exact tie-breaking rule is unimportant for our results. 

Observe that if the designer sends an uninformative signal and either $q_\ell>0.5$ or $\lambda<0.5<q_h$, then $A$ is chosen in state $\theta_A$ and $B$ in state $\theta_B$. (In the complimentary case where either $q_h=0.5$ or $q_\ell=0.5<\lambda$, outcome $A$ is chosen in both states of the world, due to our tie-breaking rule.) Furthermore, due to the population being a continuum and signals being private and conditionally independent, the outcome of the vote is deterministic, conditional on the state of the world. This means that it is either the case that the designer  \emph{cannot manipulate the election's outcome} (in which case $A$ is still chosen only in state $\theta_A$, regardless of the designer's signal), or that he can \emph{perfectly manipulate} the election's outcome, such that $A$ is chosen in both states.
In the latter case, we say that a designer's signal $s$ is \emph{optimal} if it induces the choice of policy $A$ in both states.
 Note that there might be multiple optimal signals.

\section{Results}\label{sec:results}

\subsection{Preliminaries}\label{sec:preliminaries}
Suppose a voter's belief that the state is $\theta_A$ is equal to $1/2$. Then any additional signal $s:\Theta\mapsto\{\fa,\fb\}$ induces a pair of posteriors $(\alpha, \beta)$, where
$\alpha = \p{\theta_A|\fa}$ and $\beta=\p{\theta_A|\fb}$. The signal is uninformative if $\alpha=\beta=0.5$.
If the signal is informative we can assume without loss of generality that $0\leq\beta<\frac{1}{2}<\alpha\leq1$, and in this case there exists
a (unique) signal that induces these posteriors. 
Thus, we can identify 
any informative signal with the pair of posterior distributions it induces when starting with prior $\frac{1}{2}$.\footnote{Representing signals as distributions over posteriors is standard in Bayesian persuasion \citep{kamenica2019bayesian}. Representing them as posteriors starting with a prior of $\frac{1}{2}$ is common in the literature on social learning \citep[see, e.g.,][]{acemoglu2011bayesian,arieli2020identifiable,arieli2021herd}.}
For any $\alpha<\frac{1}{2}<\beta$, the exact mapping can be derived as follows:
First, since the expected prior is equal to the posterior, we have
\begin{align*}
\alpha p_\fa &+ \beta (1-p_\fa) = \frac{1}{2}
~~~~\Rightarrow ~~~~p_\fa = \frac{\frac{1}{2}-\beta}{\alpha-\beta}, ~~~~ p_\fb = \frac{\alpha-\frac{1}{2}}{\alpha-\beta},
\end{align*}
where $p_\fa = \p{\fa}$ and $p_\fb=\p{\fb}=1-p_a$.
Next, we have that
$$\p{s(\theta_A)=\fa}=\p{\fa|\theta_A} = \frac{\p{\theta_A|\fa}p_\fa}{\p{\theta_A}} = \frac{\alpha p_\fa}{1/2} = \frac{\alpha - 2\alpha \beta}{\alpha-\beta}~~ \Rightarrow ~~ \p{\fb|\theta_A}=\frac{2\alpha\beta-\beta}{\alpha-\beta}.$$
Similarly,
$$\p{s(\theta_B)=\fa}=\p{\fa|\theta_B} = \frac{\p{\theta_B|\fa}p_\fa}{1/2} = \frac{1-2\beta-\alpha+2\alpha\beta}{\alpha-\beta}~~\Rightarrow ~~\p{\fb|\theta_B}=\frac{2\alpha+\beta-1-2\alpha\beta}{\alpha-\beta}.$$

Furthermore, this analysis implies that the bias of the signal $(\alpha, \beta)$ is:
\begin{equation}\label{eq:bias-alpha-beta}
\bias(\alpha,\beta)=\p{s(\theta_B)=\fa}-\p{s(\theta_A)=\fb}=\frac{1-\beta-\alpha}{\alpha-\beta}.   \end{equation}

Next, we have the following lemma, which allows us to focus in our analysis on a small set of $2\times 3$ signals, such that, if the election is manipulable, one of these signals must be optimal (and any other optimal signal has weakly higher values of $\alpha$ and $\beta$).
\begin{lemma}\label{lem:opt-signals-general} 
Suppose that signal  $(\alpha', \beta')$ is optimal. Then the signal $(\alpha,\beta)$ is also optimal, where $\alpha=\max\big\{q\in \{q_\ell,q_h\}:q\leq \alpha'\big\}$ and $\beta=\max\big\{q\in \{0,1-q_h, 1-q_\ell\}:q\leq \beta'\big\}$.
\end{lemma}
The intuition for Lemma \ref{lem:opt-signals-general} is that for any signal, slightly decreasing either $\alpha$ or $\beta$ increases the probability of the $A$-favorable realization $\fa$, and that if $\alpha\notin\{q_\ell,q_h\}$ and $\beta\notin\{0,1-q_h, 1-q_\ell\}$, then this slight decrease has no 
effect on the behavior of any voter conditional on her signals. See Appendix \ref{proof-Lemma-possible} for the proof. Table \ref{tab:6-signals} summarizes the key properties of these $2\times 3$ optimal signals (which are derived from straightforward calculations). Note that the share of voters who vote for policy $A$ in state $\theta_A$ is always larger than in state $\theta_B$, and, thus, a signal is optimal if and only if it induces at least half of the agents to vote for policy $A$ in state $\theta_B$.

\begin{table}[h]
\caption{Key Properties of the $2\times3$ 
Signals of Lemma \ref{lem:opt-signals-general}}\label{tab:6-signals}
\smallskip

\scalebox{0.9}{\begin{tabular}{|>{\centering}p{2cm}|>{\centering}p{2.2cm}|>{\centering}p{2.4cm}|>{\centering}p{5cm}|>{\centering}p{5cm}|}
\hline 
$\left(\alpha,\beta\right)$ & $\p{s(\theta_{B})=\mathfrak{b}}$ & Bias & How voters vote & Share of $B$ voters in $\theta_{B}$\tabularnewline
\hline 
\hline 
$\left(q_\ell,0\right)$ & $2-\frac{1}{q_\ell}$ & $\begin{array}{c}
\\
\\
\end{array}\frac{1-q_\ell}{q_\ell}>0\begin{array}{c}
\\
\\
\end{array}$ & $q_{h}$: Vote $A$ iff $a$ AND $\mathfrak{a}$
\smallskip

$q_\ell$: Follow designer's signal & $1-\frac{\left(1-q_\ell\right)\left(1-q_{h}+\lambda q_{h}\right)}{q_\ell}$
\smallskip

(decreases in $\lambda$) \tabularnewline
\hline 
$\left(q_{h},0\right)$ & $2-\frac{1}{q_{h}}$ & $\begin{array}{c}
\\
\\
\end{array}\frac{1-q_{h}}{q_{h}}>0\begin{array}{c}
\\
\\
\end{array}$ & Follow designer's signal & $\frac{2q_{h}-1}{q_{h}}$
\smallskip

(independent of $\lambda$)\tabularnewline
\hline 
$\left(q_\ell,1-q_\ell\right)$ & $\begin{array}{c}
\\
\\
\end{array}q_\ell\begin{array}{c}
\\
\\
\end{array}$ & $\begin{array}{c}
\\
\\
\end{array}0\begin{array}{c}
\\
\\
\end{array}$ & $q_{h}$: Follow original signal
\smallskip

$q_\ell$: Vote $A$ iff $a$ OR $\mathfrak{a}$ & $q_{h}-\lambda\left(q_{h}-q_\ell^{2}\right)$
\smallskip

(decreases in $\lambda$) \tabularnewline
\hline 
$\left(q_\ell,1-q_{h}\right)\begin{array}{c}
\\
\\
\end{array}$ & $\frac{q_{h}(2q_\ell-1)}{q_{h}+q_\ell-1}$ & $\frac{q_{h}-q_\ell}{q_\ell+q_{h}-1}>0$ & $q_{h}$: Follow original signal
\smallskip

 $q_\ell$: Follow designer's signal & $\lambda\frac{q_{h}(2q_\ell-1)}{q_{h}+q_\ell-1}+\left(1-\lambda\right)q_{h}$
\smallskip

(decreases in $\lambda$) \tabularnewline
\hline 
$\left(q_{h},1-q_\ell\right)\begin{array}{c}
\\
\\
\end{array}$ & $\frac{q_\ell(2q_{h}-1)}{q_{h}+q_\ell-1}$ & $-\frac{q_{h}-q_\ell}{q_\ell+q_{h}-1}<0$ & Vote $A$ iff $a$ OR $\mathfrak{a}$ & $\left(\lambda q_\ell+\left(1-\lambda\right)q_{h}\right)\frac{q_\ell(2q_{h}-1)}{q_{h}+q_\ell-1}$
\smallskip

(decreases in $\lambda$) \tabularnewline
\hline 
$\left(q_{h},1-q_{h}\right)\begin{array}{c}
\\
\\
\end{array}$ & $q_{h}$ & 0 & $q_{h}$: Vote $A$ iff $a$ OR $\mathfrak{a}$
\smallskip

 $q_\ell$: Follow designer's signal & $q_{h}^{2}+\lambda q_{h}\left(1-q_{h}\right)$
 \smallskip

(increases in $\lambda$) \tabularnewline
\hline 
\end{tabular}
}
\end{table}

\begin{remark}\label{rem-tie-breaking}
The tie-breaking rule has no significant impact on our results. In particular, if one assumes the opposite, $B$-favorable tie-breaking rules (namely, each agent votes for $B$ when indifferent, and $B$ is chosen if supported by exactly half of the voters), then the only minor effect on our results will be that the optimal signals will be perturbed by a small $\varepsilon>0$. 
Formally, if the designer can manipulate the elections, then there would be an optimal signal of the form
$(\frac{1}{2}+\varepsilon,\beta+\varepsilon)$ or
$(\alpha+\varepsilon,\beta+\varepsilon)$, with
$\alpha\in\{q_\ell,q_h\}$ and $\beta\in\{0,1-q_h, 1-q_\ell\}$, for a sufficiently small $\varepsilon>0$.
\end{remark}

\subsection{Partially Uninformed Populations ($q_\ell=\frac{1}{2}$)}\label{sec:more-general-case}
In this section we focus on the case in which the low-accuracy signal is uninformative (i.e., $q_\ell=\frac{1}{2}$). That is, we assume that a share $\lambda$ of the population is uninformed, and the remaining agents obtain each a private signal with accuracy $q_h$. An interesting special case, discussed towards the end of the section, is that of homogeneous populations, in which all voters have the same signal accuracy (formally, $\lambda=0$).

Our first observation is that if the signal accuracy $q_h$ is below some threshold, denoted by $q_{NI}(\lambda)$, then policy $A$ is always chosen, even without any manipulation by the designer.\footnote{
 If one changes the tie-breaking rule to favor policy $B$, then under the same condition on $\lambda$, the designer can perfectly manipulate the election's outcome by sending  the  signal $(\frac{1}{2}+\varepsilon,1-q_h)$ for a sufficiently small $0<\varepsilon\ll1$. This signal is almost-uninformative in the sense that with an arbitrarily  high probability of $1-O(\varepsilon)$ the agents get the realization $\fa$, which only slightly increases the likelihood of state $\theta_A$, and which is sufficient to make the uninformed voters to strictly prefer policy $A$.} 

\begin{claim}\label{claim-1}
Policy $A$ is chosen in both states with an uninformative designer's signal  iff $$q_h\leq q_{NI}(\lambda)\equiv\frac{1}{2(1-\lambda)}.$$
\end{claim}
\begin{proof}
In either state, uninformed voters always vote for $A$, and informed voters vote for $A$ if they receive signal $a$. In state $\theta_A$, the latter make up a $q_h>1/2$ share, and so overall a majority of voters always votes for $A$.
In state $\theta_B$, voters for $A$ consist of the $\lambda$-fraction of uninformed voters and the $(1-\lambda)\cdot(1-q_h)$-fraction of informed voters who receive signal $b$. Thus, at least half of the voters vote for policy $A$ in state $\theta_B$ if and only if $\lambda+(1-\lambda)\cdot(1-q_h)\geq\frac{1}{2}\Leftrightarrow \lambda\geq1-\frac{1}{2q_h}\Leftrightarrow q_h\leq\frac{1}{2(1-\lambda)}.$
\end{proof}

It is straightforward to verify that $q_{NI}(\lambda)$ is increasing in $\lambda$, that $q_{NI}(25\%)=\frac{2}{3}$, and that $q_{NI}(50\%)=1$. This is illustrated by the solid (red) curve in Figure \ref{fig-partially-informed}.
In what follows we focus on the more interesting case in which $q_h$ is above this threshold---namely, the case in which policy $B$ is chosen in state $\theta_B$ unless the designer manipulates the outcome. 

In order to state our result, we define $\overline q:[0,1)\Rightarrow[0.5,1]$ as follows:
  $$\overline q (\lambda)= \frac{-\lambda + \sqrt{\lambda^2+2-2\lambda}}{2-2\lambda}.$$
 It is straightforward to verify the following, illustrated by the dashed (blue) line in Figure  \ref{fig-partially-informed}:
 \begin{fact}\label{fact-overline-q}
 $\overline q(\lambda)$ is decreasing in $\lambda$, $\overline q(0)=\frac{\sqrt2}{2}$, and $\overline q(25\%)=\frac{2}{3}=q_{NI}(25\%)$.
  \end{fact}
 
 Our first result shows that the designer can manipulate the election's outcome iff $q_h\leq \overline q(\lambda)$, and that manipulation (when possible) can be implemented by an unbiased signal. Moreover, for values of precision in the interval $(\frac{2}{3},\overline q (\lambda)$
 (values above the horizontal dotted (black) line in Figure \ref{fig-partially-informed}), all optimal signals have non-positive biases. 
 \begin{figure}
\includegraphics[scale=0.45]{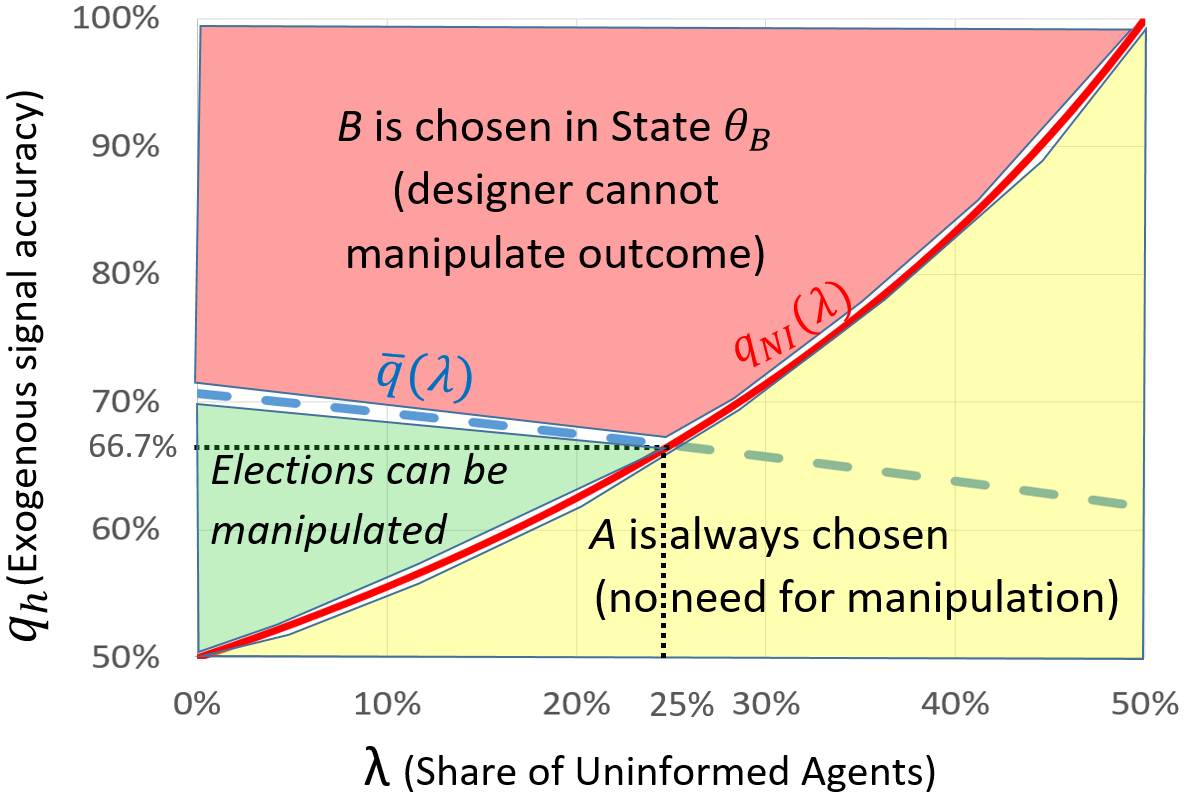}
\label{fig-partially-informed}
\caption{Illustration of Proposition \ref{prop:partially-uninformed} ($q_\ell=0.5)$}
\end{figure}

\begin{proposition}\label{prop:partially-uninformed} Suppose that $q_\ell=\frac{1}{2}$ and $q_h>q_{NI}(\lambda).$
\begin{enumerate} 
\item If $q_h>\overline q(\lambda)$ then the designer cannot manipulate the election's result.
\item If $q_h\leq\overline q(\lambda)$ then the designer can manipulate the election's outcome, and can do so by using the unbiased signal $(q_h,1-q_h)$.
\item If $q_h\in(\frac{2}{3},\overline q(\lambda)]$ then {\em all} optimal signals have non-positive biases.
\end{enumerate}
\end{proposition}

An interesting feature of the optimal signals is that they have non-positive bias (bullets 2 and 3 of Proposition~\ref{prop:partially-uninformed}). Observe that this is the opposite direction of bias relative to the standard application of Bayesian persuasion \citep{kamenica2011bayesian}, in which the bias is positive (i.e., in favor of $A$, the designer-optimal outcome). The intuition underlying the contrast between the standard application and our setup is the following. A voter who is informed and obtains signal $a$ is  prepared to vote for $A$, whereas a voter who is informed and obtains signal $b$ is prepared to vote for $B$. The designer would like to convince the latter to also vote for $A$, while at the same time not causing the former to switch. The  way to do this is to choose a signal for which realization $\fa$ is stronger than realization $\fb$, so that voters with realizations $(b,\fa)$ switch to $A$ but voters with realizations $(a,\fb)$ do {\em not} switch to $B$. Of course, since $\fa$ is stronger than $\fb$, it must be the case that $\fb$ is realized with higher probability, or, in other words, that the signal has non-positive bias. We now prove Proposition~\ref{prop:partially-uninformed}.
\begin{proof}
Claim \ref{claim-1} and the inequality $q_h>q_{NI}(\lambda)$ imply that policy $A$ is not chosen in state $\theta_B$ if the signal is uninformative. Due to this we can focus on informative signals. Thus, Lemma \ref{lem:opt-signals-general} implies that if the designer can manipulate the election, then either $(q_h,0)$ or $(q_h,1-q_h)$ are optimal. Signal $(q_h,0)$ induces all voters to vote according to the designer's signal $s$ (and to ignore the original signal). The probability of signal $s$ being in favor of $A$ in state $\theta_B$ is $\Pr (s(\theta_B)=\fa)=\frac{1-2\beta-\alpha+2\alpha\beta}{\alpha-\beta}=\frac{1-q_h}{q_h},$ which is weakly larger than $50\%$ iff $q_h\leq\frac{2}{3}.$ 

Signal $(q_h,1-q_h)$ induces informed (resp., uninformed) agents to vote for $A$ if either of their two signals (resp., the designer's signal) is in favor of $A$, which has probability of $1-q_h^2$ (resp., $1-q_h$) in state $\theta_B$. Thus, the share of voters who vote for policy $A$ in state $\theta_B$ is at least  $50\%$ iff $\lambda\cdot(1-q_h)+(1-\lambda)\cdot(1-q_h^2)\geq50\% \Leftrightarrow q\leq \overline q(\lambda).$ Combining these observartions with Fact \ref{fact-overline-q} implies bullets (1) and (2) of the proposition. 

The argument for bullet (3) is the following. Due to Lemma \ref{lem:opt-signals-general}, if $\frac{2}{3}<q_h\leq \frac{\sqrt2}{2}$, then any optimal signal $(\alpha,\beta)$ must satisfy $\alpha\geq{q_h}$ and $\beta\geq{1-q_h}$.
In this case, the numerator in (\ref{eq:bias-alpha-beta}), namely, $1-(\alpha+\beta)$, is non-positive. Since the denominator is always positive, this implies that the bias is non-positive.
\end{proof}
An interesting special case is a homogeneous population in which all agents have the same signal quality (i.e., $\lambda=0$). Applying Proposition \ref{prop:partially-uninformed} yields the following characterization of election manipulation in homogeneous populations.
\begin{corollary}[Homogeneous populations]\label{cor-homogenous}
Suppose that $\lambda=0$.
\begin{enumerate} 
\item If $q_h>\frac{\sqrt2}{2}$ then the designer cannot manipulate the election's result.
\item If $q_h\leq\frac{\sqrt2}{2}$ then the designer can manipulate the election's outcome, and can do so by using the unbiased signal $(q_h,1-q_h)$.
\item If $q_h\in(\frac{2}{3},\frac{\sqrt2}{2}]$ then {\em all} optimal signals have non-positive biases.
\end{enumerate}
\end{corollary}
\subsection{General Binary Signals $\left(\frac{1}{2}<q_\ell<q_h,\,\,\lambda\in[0,1]\right)$}\label{sec:most-general-case}
In this section we return to the general binary model in which the low-accuracy signal is informative ($q_\ell>1/2$). We prove two results: first, a characterization of conditions under which the designer can manipulate the election; and second, an identification of conditions under which optimal signals have positive or negative biases.

In order to state our first result we define $\underline\lambda:(\frac{1}{2},\frac{\sqrt2}{2})]\rightarrow[0,1]$ as follows:
$\underline\lambda(q_h)=\frac{0.5-q_{h}^{2}}{q_{h}\left(1-q_{h}\right)}.$ We will show in the proof of Proposition \ref{pro-general-binary-signals} below that $\underline\lambda(q_h)$ is the highest share of $\lambda$ for which the signal $(q_h,1-q_h)$ induces a majority of the voters to vote for $A$ in state $\theta_B$ for all values of $q_\ell\leq q_h$.
Observe that:
\begin{fact}\label{fact-lambda-under}
The function $\underline\lambda(q_h)$ is decreasing, $\underline\lambda(\frac{2}{3})=25\%$, and $\underline\lambda(\frac{\sqrt2}{2})=0$.
\end{fact}
\begin{proposition}\label{pro-general-binary-signals}~
\begin{enumerate}
    \item If $\frac{\sqrt2}{2}<q_\ell$, then the designer cannot manipulate the election.
    \item If $q_\ell<\frac{\sqrt2}{2}<q_h$, then the designer can manipulate the election iff $\lambda$ is sufficiently large.
    \item If $\frac{2}{3}<q_h\leq\frac{\sqrt2}{2}$, then the designer can manipulate the elections iff either (1) $\lambda\leq\underline\lambda(q_h)$, or (2) $\lambda$ is sufficiently large.
    \item If $q_h\leq\frac{2}{3}$, then the designer can manipulate the elections for any $q_\ell$ and  $\lambda$.
    
\end{enumerate}
\end{proposition}
\begin{proof}~
\begin{enumerate}
    \item Corollary \ref{cor-homogenous} implies that for each $q>\frac{\sqrt2}{2}$ all signals induce less than $50\%$ of agents with signal accuracy $q$ to vote for policy $A$ in state $\theta_B.$ This implies that if $q_\ell>\frac{\sqrt2}{2},$ then a majority of the voters support policy $B$ in  state $\theta_B$ (i.e., the designer cannot manipulate the elections).
    \item This is implied by the observation that when $\frac{\sqrt2}{2}<q_h$ both signals $(q_h,0)$ and $(q_h,1-q_h)$ cannot manipulate the elections, while the share of agents who support policy $A$ in state $\theta_B$ is increasing in $\lambda$ (and is greater than than 50\% for sufficiently high $\lambda$) for each of the other optimal signals of Lemma \ref{lem:opt-signals-general} (see Table \ref{tab:6-signals}). Signal $(q_h,0)$  cannot manipulate the elections because the induced share of $B$ supporters in state $\theta_B$ is $\frac{2q_h-1}{q_h},$  which is larger than $50\%$ iff $q_h>\frac{\sqrt2}{2}$.  Signal $(q_h,1-q_h)$ induces a share $q_h^2+\lambda q_h(1-q_h)$ of the voters  to vote for policy $B$ in state $\theta_B$ (see Table \ref{tab:6-signals}). This share is smaller than $50\%$ iff $\lambda\leq\underline\lambda(q_h)$ because $q_h^2+\lambda q_h(1-q_h)\leq0.5 \Leftrightarrow \lambda \leq \frac{0.5-q_h^2}{q_h(1-q_h)}\equiv \underline\lambda(q_h).$ As observed in Fact \ref{fact-lambda-under}, $\underline\lambda(\frac{\sqrt2}{2})=0,$ which implies that signal $(q_h,1-q_h)$ cannot manipulate the elections if $q_h>\frac{\sqrt2}{2}.$ 
    \item Part (1) is implied by the observation that the designer's signal $\left(q_h,1-q_h\right)$ induces a share $q_h^2+\lambda q_h(1-q_h)$ to vote for policy $B$ in state $\theta_B$, and that this share is less than $50\%$ if $\lambda\leq\underline\lambda(q_h).$ Part (2) is implied by the fact that signal $(q_h,0)$ (resp., $(q_h,1-q_h)$) does not manipulate the elections when $q_h>\frac{2}{3}$ (resp., $\lambda>\underline\lambda(q_h)$), and that the share of supporters of policy $A$ in state $\theta_B$ is increasing in $\lambda$ for each of the other four signals of Lemma \ref{lem:opt-signals-general}.
    \item Observe (Table \ref{tab:6-signals}) that the designer's signal $\left(q_h,0\right)$ induces a share of $\frac{2q_h-1}{q_h}$ of the agents to vote for $B$ in state $\theta_B$ (for all values of $q_\ell\leq q_h$), and that this share is smaller than $50\%$ iff $q_h\leq \frac{2}{3}.$ \qedhere
\end{enumerate}
\end{proof}

Our results are illustrated in Figure \ref {fig:q_l-lambda}. Bullet (3) of Proposition \ref{pro-general-binary-signals} implies an interesting non-monotonicity in $\lambda$: When $q_h$ is in the interval $\left(\frac{2}{3},\frac{\sqrt2}{2}\right)$, then the elections are manipulable if either the share of agents with low accuracy is sufficiently small (namely, $\lambda\leq\underline\lambda$) or if it is sufficiently large, but the elections cannot be manipulated if this share is intermediate. The intuition for this non-monotonicity is as follows. Our designer is constrained by the fact that he must choose the same signal distribution to all agents. When $\lambda$ is either very low or very high, the designer can tune this distribution to the large majority of agents by sending an unbiased signal with the ``right'' level of accuracy (i.e., a level that induces most voters to vote $A$ if either of their signals is in favor of $A$). By contrast, when $\lambda$ is close to $50\%$, the designer has a trade-off, and must send a ``versatile'' signal that can handle both groups of agents (and is not optimally ``tailored'' to either group), and this limits his ability to manipulate the elections.

Moreover, one can show that the ability to manipulate the elections is also non-monotonic in $q_\ell$. Specifically, for each $q_h$, the interval of $\lambda$-s for which the elections cannot be manipulated is non-monotone in $q_\ell$: the interval is small (and it might even disappear) when $q_\ell$ is either small (close to $50\%$) or large (close to $q_h$), while it is largest when $q_\ell$ is intermediate. The intuition for the non-monotonicity in $q_\ell$ is as follows. When $q_\ell$ is low, it is easy to manipulate the low-accuracy agents. The designer can send a positively biased signal that with high probability sends a ``weak-pro-$A$'' realization, which manipulates most low-accuracy agents (and a sufficient number of high-accuracy agents). When $q_\ell$ is close to $q_h$, the designer can exploit the relative homogeneity of the population by sending the negatively biased signal $(q_h,1-q_\ell)$, which induces both types of agents to vote for $A$ if either of their signals is ``pro-$A$''. In contrast, when $q_\ell$ has an intermediate value the designer has neither of these two benefits, and thus his ability to manipulate the elections is limited.

\begin{figure}[h]
\begin{centering}
\includegraphics[scale=0.44]{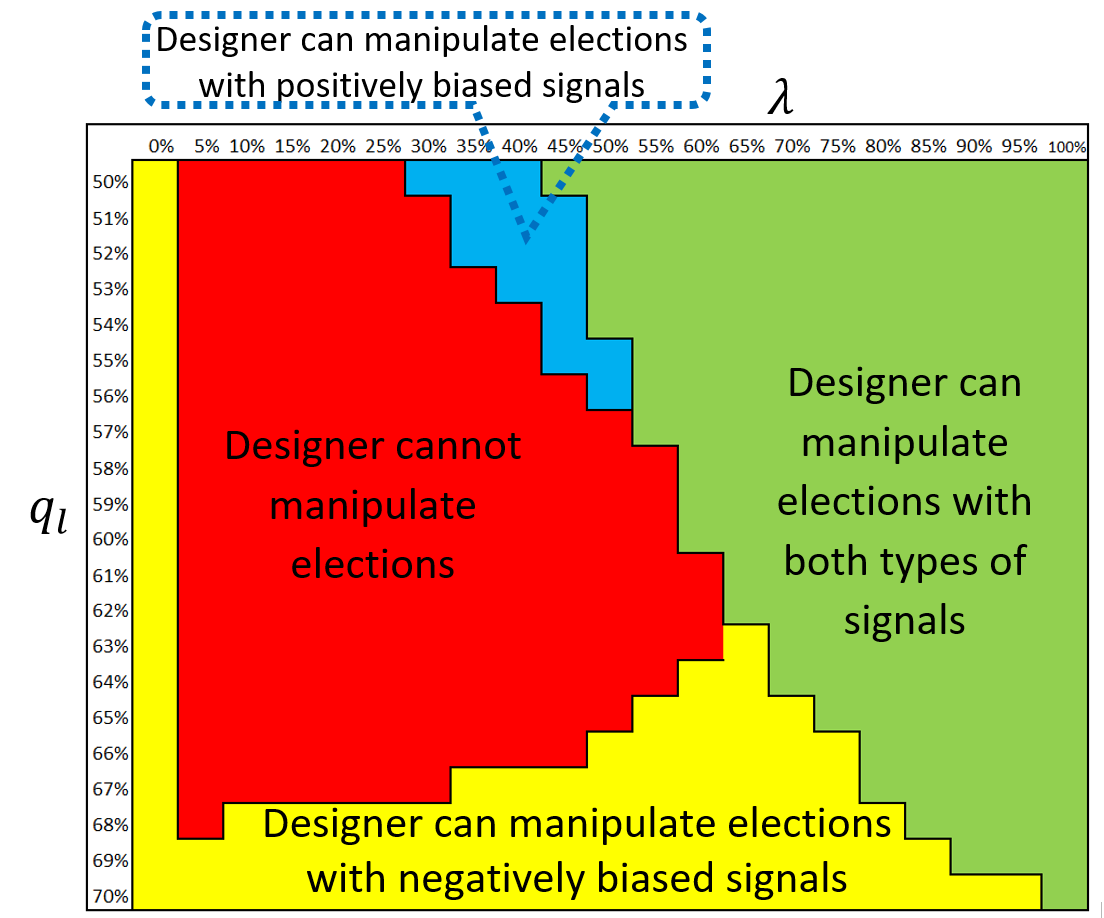}
\par\end{centering}
\caption{Outcome Manipulation for $q_{h}=70\%$ and Various Values of $q_{\ell}$ (values between $50\%$ and $70\%$ in steps of $1\%$) and $\lambda$ (values between 0 and $100\%$ in steps of $5\%$)}
\label{fig:q_l-lambda}
\end{figure}

Recall that in the special case of $q_\ell=0.5$ (partially informed agents), manipulation, when possible, could always be implemented by an unbiased signal (and in some cases only by signals with non-positive biases). In what follows, we show that this is no longer the case in the general setup (with $q_\ell>0.5$). Specifically, we characterize both setups in which manipulation can only be implemented by positively biased signals (in which the signal $\mathfrak{a}$ is sent more often than $\mathfrak{b}$), and setups in which manipulation can only be implemented by negatively biased signals. Specifically, Proposition \ref{pro-direction-bias} below shows that the direction of the bias is uniquely determined (for an interior interval of $\lambda$-s) in the following setups:
\begin{enumerate}
    \item All optimal signals must have a  positive bias if $q_h>\frac{2}{3}$ and $q_\ell$ is sufficiently low.  
    \item All optimal signal must have a  negative bias if $q_h\in(\frac{2}{3},\frac{\sqrt2}{2})$ and $q_\ell$ is sufficiently high.
\end{enumerate}
Observe that in the second case, the manipulation is implemented by signals with the opposite direction of bias relative to the standard applications of Bayesian persuasion (see, e.g., \citealp{kamenica2019bayesian}). The intuition for this is that the optimal negatively biased signal $(q_h,1-q_\ell)$ exploits the (modest) heterogeneity of the population: the relatively infrequent signal $\mathfrak{a}$ is sufficiently informative to induce both types of agents to vote $A$ regardless of their exogenous signal, while the more frequently used signal $\mathfrak{b}$ is sufficiently weakly-informative, such that it does not change the voting of any agent who got exogenous signal in favor of policy $A$.
\begin{proposition}\label{pro-direction-bias}
~
\begin{enumerate}
    \item For any $q_h>\frac{2}{3}$ there is a $\overline q\in(0.5,q_h)$ such that, if $q_\ell\in(0.5, \overline q)$, then there is an interior open  interval of $\lambda$-s for which all optimal signals have positive biases.
    \item  For any $q_h\in\left(\frac{2}{3},\frac{\sqrt2}{2}\right)$ there is a  $\underline q\in(0.5,q_h)$, such that, if $q_\ell\in(\underline q,q_h)$, then there is an interior open  interval of $\lambda$-s for which all optimal signals have negative biases.
\end{enumerate}
\end{proposition}
\begin{proof}  Fix $q_h>\frac{2}{3}.$ Proposition \ref{pro-general-binary-signals} and Fact \ref{fact-lambda-under} imply that the signals $(q_h,0)$ and $(q_h,1-q_h)$ are not optimal if $\lambda>0.25-\epsilon$ for a sufficiently small $\epsilon>0.$ 
\begin{enumerate}
\item The limit of the share of $B$ voters in state $\theta_B$ when $q_\ell$ converges to 0.5 (see Table \ref{tab:6-signals}) is equal to $q_h(1-\lambda)$ for the positively biased signals $(q_\ell,0)$ and $(q_\ell,1-q_h)$, while it is strictly higher for the remaining two signals: it is equal to $q_h(1-\lambda)+\frac{1}{4}$ for signal $(q_\ell,1-q_\ell)$, and it is equal to $q_h(1-\lambda)+\frac{\lambda}{2}$ for signal $(q_h,1-q_\ell)$. This implies that for sufficiently small $q_\ell$-s, if $\lambda$ is slightly above $1-\frac{1}{2q_h}$, then only the positively biased signals $(q_\ell,0)$ and $(q_\ell,1-q_h)$ are optimal in the set of six signals of Lemma \ref{lem:opt-signals-general}. By Lemma \ref{lem:opt-signals-general}, this implies that the only potentially optimal signals are $(\alpha',\beta')\in[q_\ell,q_h)\times [0,1-q_h)$ or $(\alpha',\beta')\in[q_\ell,q_h)\times [1-q_h,1-q_\ell)$. Note that the former are all positively biased. The latter, however, may be either positively or negatively biased.

We now show that only the positively biased signals $(\alpha',\beta')\in[q_\ell,q_h)\times [1-q_h,1-q_\ell)$ are potentially optimal. Fix some $(\alpha',\beta')\in[q_\ell,q_h)\times [1-q_h,1-q_\ell)$, and define $q_\ell'$ and $q_h'$ so that $(\alpha',\beta')=(q_\ell', 1-q_h')$. Observe that this signal is positively biased if $q_\ell'<q_h'$, and negatively biased if the inequality is flipped. Under this signal, the fraction of voters for $B$ in state $\theta_B$ is $$\p{B|\theta_B}=\lambda\frac{q_{h}'(2q_\ell'-1)}{q_{h}'+q_\ell'-1}+\left(1-\lambda\right)q_{h}.$$
Now, if the signal is negatively biased, and so $q_\ell'\geq q_h'$, then $2q_\ell'-1\geq q_h'+q_\ell'-1$, and so the first term of $\p{B|\theta_B}$ is bounded below by $\lambda q_h' \geq \lambda/2$. The total fraction of voters for $B$ is then at least $q_h(1-\lambda)+\lambda/2$, the same as for signal $(q_h,1-q_\ell)$. But just like that signal, if $\lambda$ is slightly above $1-\frac{1}{2q_h}$ then a negatively biased signal $(\alpha',\beta')\in[q_\ell,q_h)\times [1-q_h,1-q_\ell)$ is also not optimal.

\item Assume $q_h\in(\frac{2}{3},\frac{\sqrt2}{2}).$  The limit of the share of $B$ voters in state $\theta_B$ when $q_\ell$ converges to $q_h$ (see Table \ref{tab:6-signals}) for each of the other four signals is equal to: 
\begin {enumerate}
\item Signal $(q_\ell,0)$: The limit is equal to $1-\frac{\left(1-q_{h}\right)\left(1-(1-\lambda)q_{h}\right)}{q_{h}}$, which is increasing in $q_h$ and decreasing in $\lambda$. When substituting $q_h=\frac{2}{3}$ and  $\lambda=1$, the expression is equal $50\%$, which implies that it is larger than $50\%$ for each $q_h>\frac{2}{3}$ and each $\lambda$. 
\item Signal $(q_\ell,1-q_\ell)$: The limit is equal to $q_{h}\left(1-\lambda(1-q_{h}\right))$, which is increasing in $q_h$ and decreasing in $\lambda$. When substituting $q_h=\frac{2}{3}$ and  $\lambda=\frac{2}{3}$, the expression is equal $\frac{14}{27}>50\%$, which implies that the limit is larger than $50\%$ for each $q_h>\frac{2}{3}$ and each $\lambda<\frac{2}{3}$. 
\item Signal $(q_\ell,1-q_h)$: The limit is equal to $q_h>\frac{2}{3}$,  which implies that the limit is larger than $50\%$ for each $q_h>\frac{2}{3}$ and each $\lambda$. 
\item Signal $(q_h,1-q_\ell)$: The limit is equal to $q_{h}^{2}$, which is smaller than $50\%$ for all values of $\lambda$ because $q_h<\frac{\sqrt2}{2}.$
\end{enumerate}
This implies that when $q_\ell$ is sufficiently close to $q_h$ and $\lambda\in(0.25,\frac{2}{3})$, then the unique optimal signal in the set of six signals of Lemma \ref{lem:opt-signals-general} is $(q_h, 1-q_\ell)$, which implies that all optimal signals are negatively biased due to  Lemma \ref{lem:opt-signals-general}. \qedhere
\end{enumerate}
\end{proof}

\section{Extensions}\label{sec:extensions}
\subsection{Continuous Signal Accuracies}\label{sec:continuous}
In this section we extend our model to a general distribution over signal accuracies (and not just binary). Specifically, we consider here a variant of our model in which the exogenous signal's accuracy is continuous, and is distributed according to density $f$ with a support 
$\supp(f)$ in $[0.5,1]$. 
We note that the closely related extension of the model to a setup with a finite (but greater than 2) number of signal accuracies is analogous (and is omitted for brevity).

Our result here is an extension of Proposition~\ref{pro-general-binary-signals} to this general setup, showing that the election is manipulable if all agents have accuracies of at most $\frac{2}{3}$, and that it is non-manipulable if all agents have accuracies of at least $\frac{\sqrt 2}{2}$. Formally,
\begin{proposition}
~
\begin{enumerate} 
\item If $\supp(f) \subseteq \left[\frac{\sqrt{2}}{2}, 1\right]$, then the designer cannot manipulate the election.
\item If $\supp(f)\subseteq \left[\frac{1}{2}, \frac{2}{3}\right]$, then the designer can manipulate the election.
\end{enumerate}
\end{proposition}
\begin{proof}
~
\begin{enumerate}
    \item Corollary \ref{cor-homogenous} implies that for each $q>\frac{\sqrt2}{2}$ all signals induce less than $50\%$ of agents with signal accuracy $q$ to vote for policy $A$ in state $\theta_B.$ This implies that when $\supp(f) \subseteq \left[\frac{\sqrt{2}}{2}, 1\right]$,  a majority of the voters support policy $B$ in  state $\theta_B$ irrespective of the exact distribution of signal accuracies (i.e., the designer cannot manipulate the elections).
      \item Let $q_h = \sup_{q\in Q} q$. Observe (Table \ref{tab:6-signals}) that the designer's signal $\left(q_h,0\right)$ induces a share of $\frac{2q_h-1}{q_h}$ of the agents to vote for $B$ in state $\theta_B$ (for all values of
      $q\leq q_h$), and that this share is smaller than $50\%$ iff $q_h\leq \frac{2}{3}.$ \qedhere
    \end{enumerate}
\end{proof}

\subsection{Targeted Information}\label{sec:targeted}

In our main analysis, we required the designer to design one signal for all voters. What if the designer can differentiate between voters based on their signal accuracies, and can thus design different signals for different voters? 

Suppose that the designer can distinguish between the two types of voters, $q_\ell$ and $q_h$. Furthermore, suppose she can commit to two different signals, $s_h$ and $s_\ell$, where the former targets $q_h$ voters and the latter targets $q_\ell$ voters.
Observe first that Proposition~\ref{pro-general-binary-signals} applies: if $q_\ell<q_h\leq 2/3$ then the $(q_h,0)$ signal is optimal regardless of the type of voter, and so the designer can manipulate the election, whereas if $q_h>q_\ell>\sqrt{2}/2$ then the designer cannot manipulate the election. The proof is the same as that of the proposition.

However,  targeted signals differ in the intermediate case. Note that for either $i\in\{\ell, h\}$, the signal that leads to the maximal fraction of $q_i$ voters to vote for $A$ is one of $(q_i, 0)$ or $(q_i, 1-q_i)$, by (the proof of) Lemma \ref{lem:opt-signals-general}. Both signals lead to a majority of voters for $A$ in state $\theta_A$. In state $\theta_B$, however, $(q_i, 0)$ leads to
$\p{A|\theta_B} = \frac{1-q_i}{q_i}$, whereas $(q_i, 1-q_i)$ leads to $\p{A|\theta_B} = 1-q_i^2$ (see Table~\ref{tab:6-signals}).
For $q_i\in[1/2, \frac{\sqrt{5}-1}{2})$ the former is larger, whereas for $q_i\in(\frac{\sqrt{5}-1}{2},1)$ the latter is larger.
Overall, the designer can manipulate the election if and only if
$$\lambda\cdot \max\left\{\frac{1-q_\ell}{q_\ell}, 1-q_\ell^2\right\}+(1-\lambda)\cdot \max\left\{\frac{1-q_h}{q_h}, 1-q_h^2\right\}\geq \frac{1}{2}.$$
Observe that, in contrast with Proposition~\ref{pro-general-binary-signals}, there is no non-monotonicity here: the left-hand-side of the inequality decreases with $q_\ell$ and with $\lambda$.

This last necessary and sufficient condition on the designer's ability to manipulate the election extends in a straightforward manner to the variant with a continuous distribution of signal accuracies from Section~\ref{sec:continuous} above. In this case, the designer can manipulate the election if and only if
$$\int_{1/2}^1 \max\left\{\frac{1-q}{q}, 1-q^2\right\}\cdot f(q) dq \geq \frac{1}{2}.$$

\subsection{Strongly Targeted Information}\label{sec:strongly-targeted}
What if the designer has even more finely grained information about voters, such that she knows not only each voter's signal accuracy but also the realization of this exogenous signal? In this case, the designer can provide a strongly targeted signal, in which voters with different signal accuracies {\em and different realizations} get different signals.

For a given signal accuracy $q$, a voter with realization $a$ has interim belief $q$ about the probability the state is $A$, whereas a voter with realization $b$ has interim belief $1-q$ about this probability. If the designer knows these realizations, she faces a standard Bayesian persuasion problem relative to each one of these kinds voters. In order to maximize the probability that such voters vote for $A$, the designer should supply the former voter with no additional signal, and the latter voter with additional signal $(q,0)$. The latter signal here is the optimal signal from the standard Bayesian persuasion setting \citep{kamenica2011bayesian}, supplied to a voter who has belief $1-q$ and is will to vote for $A$ once the belief is above $1/2$. Given these signals, the former voter will always vote for $A$, whereas the latter will vote for $A$ with probability $\frac{1-q}{q}$ in state $\theta_B$. In state $\theta_B$, the total fraction of voters with accuracy $q$ who vote for $A$ is thus
$$(1-q)\cdot 1 + q\cdot\frac{1-q}{q} = 2(1-q).$$

For a homogeneous population in which all agents have accuracy $q$, the designer can manipulate the elections iff $2(1-q)\geq\frac{1}{2}\Leftrightarrow q\leq \frac{3}{4}.$ When there are two signals accuracies, $q_\ell$ and $q_h$, the designer can then manipulate the election if and only if
$$2\lambda\cdot (1-q_\ell)+2(1-\lambda)\cdot (1-q_h)\geq \frac{1}{2}.$$

When there is a continuous distribution of signal accuracies, as in Section~\ref{sec:continuous} above, the designer can manipulate the election if and only if 
$$2\int_{1/2}^1 (1-q)f(q)dq\geq \frac{1}{2}.$$

\subsection{Social Media vs.\ Traditional Media}\label{sec:social-vs-traditional}
Our analysis and results lead to a straightforward comparison of the effects of social media (through {\em private} persuasion) and traditional media (through {\em public} persuasion). For the latter, suppose that instead of each voter obtaining a conditionally independent realization of the designer's signal, all voters obtain the same realization. Is such a public signal better or worse for the designer?

Observe that with a public signal, the designer can never manipulate the election with probability 1. However, she can always at least slightly increase the probability that voters vote for $A$ \citep[with a 
positively biased signal as in][]{kamenica2011bayesian}. Thus, the answer to whether this is better than private persuasion depends on whether or not the designer can manipulate the election in the latter case. If she can, then private persuasion by social media is better. If she cannot, then public persuasion by traditional media is better.

\section{Conclusion}\label{sec:conclusions}

In portraying private beliefs that contradict public information,  \citet{kuran1997private}  describes a government monopoly on information under autocracy. Monopoly on information also arises under democracy when, because of network externalities, a social media platform caters to designated users seeking to maintain contact with one another. In this paper we showed that manipulation through personalized private information in social media can influence beliefs. This, in turn, may invalidate the conclusion of Condorcet's jury theorem, a theorem that captures the benefits of majority voting. In particular, information manipulation through private persuasion can result in a majority voting for an information designer's preferred policy rather than for a socially preferred alternative. This compromise of Condorcet's jury theorem undermines the merits of democracy. 

From the perspective of theory, we focused on the direct effect of informational persuasion on voting outcomes by dispensing with the assumption, prominent in the literature on information design, that voters regard themselves as possibly pivotal. Instead, we viewed voting as sincere, and accounted for the paradox of voting by assuming that voters derive direct utility from expressing themselves through voting. Our general model covers a range of circumstances that differ in the extent to which users are exogenously informed and in the amount of heterogeneity in the accuracy levels of users' exogenous information. In all these circumstances we provided tight characterizations of when Condorcet's jury theorem can be overturned.

Our analysis and results provide a comparison between information manipulation by  private persuasion through social media and by public persuasion through traditional media. Under public persuasion, the designer can always slightly increase the probability that voters choose the designer's preferred outcome, but can never determine the outcome with probability one. Under private persuasion, in contrast, our results show that the designer either completely determines the outcome, or has no effect at all, depending on the circumstances. In the former case, a designer with a political agenda would thus prefer control of social media to information manipulation through traditional media, whereas in the latter case, the designer's preference would be reversed.

\appendix
\section {Appendix: Proof of Lemma \ref{lem:opt-signals-general} \label{proof-Lemma-possible}}

Fix any optimal $(\alpha', \beta')$, and note that, under this signal, 
$$\p{\fa|\theta_B}=1-\frac{2\alpha'+\beta'-1-2\alpha'\beta'}{\alpha'-\beta'}=\frac{(\alpha'-1)(2\beta'-1)}{\alpha'-\beta'}.$$
Observe that this probability is decreasing in $\alpha'$ and $\beta'$, and so lowering either of these leads to a higher probability of  realization $\fa$ in state $\theta_B$.

Suppose now that $\alpha'\notin\{q_\ell,q_h\}$. First, observe that $\alpha'$ cannot lie in the interval $[1/2, q_\ell)$. This is because, if it did, then when a voter obtains exogenous signal $b$, neither the pair $(b,\fa)$ nor the pair $(b,\fb)$ will cause the voter to vote $A$. This means that voters always vote $B$ on signal $b$. However, in state $\theta_B$ the signal $b$ is more likely than $a$, and so a majority will always vote for $B$ in this state. But this implies that $(\alpha',\beta')$ is not optimal, a contradiction.

Suppose next that $\alpha'\in (q_\ell, q_h)$. In this case, switching to signal $(\alpha,\beta')$, with $\alpha=q_\ell$ does not affect the behavior of any voter: Voters with signals $(a,\fb)$ and $(b,\fb)$ still make the same inferences, whereas voters with signals $(a,\fa)$ still vote for $A$, $q_\ell$ voters with signals $(b,\fa)$ still vote for $A$, and $q_h$ voters with signals $(b,\fa)$ still vote for $B$. Similarly, if $\alpha'\in (q_h, 1)$, then switching to signal $(\alpha,\beta')$, with $\alpha=q_h$ does not affect the behavior of any voter. In all cases, however, lowering $\alpha'$ to $\alpha$ increases the probability of signal $\fa$, and hence a weakly greater fraction of voters votes for $A$. This implies that $(\alpha, \beta')$ is optimal.

Finally, suppose $\beta'\notin\{0,1-q_h, 1-q_\ell\}$. If $\beta' \in (0, 1-q_h)$ (resp., $\beta' \in (1-q_h, 1-q_\ell)$ or $\beta' \in (1-q_\ell, 1)$) then lowering $\beta'$ to $\beta=0$ (resp., $\beta=1-q_h$ or $\beta=1-q_\ell$) does not affect the behavior of any voter. In all cases, however, lowering $\beta'$ to $\beta$ increases the probability of signal $\fa$, and hence a weakly larger share of voters votes for $A$. This implies that $(\alpha, \beta)$ is optimal.

\section{Appendix: Binary Signals are Optimal}\label{apx:binary-opt}
In this section we show that the restriction of the designer's signal to a binary one is without loss of generality.
\citet{kamenica2011bayesian} show that, in the standard Bayesian persuasion setting with a single, uninformed receiver, it is without loss of generality to restrict the designer to a  {\em straightforward} signal---one where each signal realization corresponds to a distinct action, these realizations are interpreted by the receiver as recommendations, and the receiver optimally follows the recommendation. 

When there are different types of receivers (in our case, ones with different exogenous-signal accuracies) who  have different information (in our case, different realizations of the exogenous signal), the argument of \citet{kamenica2011bayesian} does not directly apply. This is because different receivers would optimally take different actions even given the same signal realization, and so these realizations cannot be interpreted as recommendations. For a simple example consider homogeneous voters with signal accuracy $q\in(0.5,1)$, and where the designer's signal is $(0.5+\eps,0.5-\eps)$ for small $\eps>0$. In this case each voter votes according to their own exogenous information, and not according to the sender's ``recommended'' action $\fa$ or $\fb$.

However, an extended argument does show that binary signals are sufficient. Consider the general model of Section~\ref{sec:most-general-case}. Then:
\begin{lemma}\label{lem:binary-signal-suffices}
If there exists an optimal signal, then there exists an optimal binary signal.
\end{lemma}

\begin{proof}
In our model there are four different types of voters, distinguished by their signal accuracies ($q_\ell$ or $q_h$) and the realization of their exogenous signal ($a$ or $b$). 
We note that our argument works for any finite, arbitrarily large number of voters' types.

Now, suppose there exists an optimal signal. We first argue that there also exists an optimal signal with a finite number of realizations. This follows from the argument of \citet{kamenica2011bayesian}, which implies that there exists an optimal signal in which there are at most 16 realizations, as follows: Each realization of the signal is a recommendation to vote either for $A$ or for $B$ for each of the four different types of voters. There are $2^4$ such four-tuples of recommendations, hence 16 signal realizations.

Next, let $s$ be an optimal designer signal with the {\em minimal} number of signal realizations, and suppose towards a contradiction that $s$ is not binary.
Let $n>2$ be the number of its signal realizations, and denote these realizations by $\supp(s)=\{r_1,\ldots,r_n\}$. By Lemma~\ref{lem:split} below, the signal $s$ is equal to a distribution over two signals, 
$s_1$ and $s_2$, with the following four properties:
(1) $s\equiv \eta s_1 + (1-\eta)s_2$ with $\eta\in(0,1)$, (2) $\abs{\supp(s_1)},\abs{\supp(s_2)}<n$, (3)
$\abs{\supp(s_1)\cap \supp(s_2)}\leq 1$, and
(4) if $r_i \in \supp(s_1)\cap \supp(s_2)$ then  $\p{\theta_B|s_1=r_i}=\p{\theta_B|s_2=r_i}$.

Now, $s$ is optimal if and only if the probability it induces a random voter to vote for $A$ in state $\theta_B$ is at least $1/2$. Since $s$ is a mixture of $s_1$ and $s_2$ 
(and the unique shared realization induces the same posterior), either $s_1$ or $s_2$ must also be optimal. 
For otherwise, if in state $\theta_B$ a random voter has probability less than $1/2$ of voting for $A$ under each of $s_1$ and $s_2$, then this must also be the case under $s$. However, $s_1$ and $s_2$ have fewer realizations than $s$, contradicting the assumption that $s$ is an optimal signal with the minimal number of realizations.
\end{proof}

\begin{lemma}\label{lem:split}
Fix a non-binary signal $s$ with $n>2$ signal realizations, and denote those realizations by $\supp(s)=\{r_1,\ldots,r_n\}$. Then there exist signals $s_1, s_2$ with the following properties:
\begin{enumerate}
    \item $s\equiv \eta s_1 + (1-\eta)s_2$, where $\eta\in(0,1)$, 
    \item $\abs{\supp(s_1)},\abs{\supp(s_2)}<n.$ 
    \item $\abs{\supp(s_1)\cap \supp(s_2)}\leq 1$, and
    \item if $r_i \in \supp(s_1)\cap \supp(s_2)$ then $\p{\theta_B|s_1=r_i}=\p{\theta_B|s_2=r_i}$.
\end{enumerate}
\end{lemma}

\begin{proof}
Denote by $(\alpha_1,\ldots,\alpha_n)$ the posteriors induced by the $n$ signal realizations of $s$ on state $\theta_A$, starting with prior $1/2$. This is simply the generalization of the $(\alpha,\beta)$ notation for binary signals to signals with more than two realizations. Suppose $r_i$ is the signal realization associated with posterior $\alpha_i$, and denote by $p_i=\p{s=r_i}$ (the {\em unconditional} probability of this realization).

If signal $s$ is non-informative, then the result is immediate (as each $s_i$ can be a non-informative signal with a single realization). Otherwise, there must be at least one $\alpha_i$  above $1/2$ and at least one $\alpha_j$ below $1/2$.  Without loss of generality, assume that $\alpha_1>\ldots>\alpha_n$, and note that $\alpha_1>1/2$ and $\alpha_n<1/2$. 

Consider the binary signal $s_1=(\alpha_1,\alpha_n)$ with $\supp(s_1)=\{r_1,r_n\}$. For this signal, let 
$$p_{\fa}=\p{s_1=r_1}=\frac{\frac{1}{2}-\alpha_n}{\alpha_1-\alpha_n}~\mbox{ and }~p_{\fb}=\p{s_1=r_n}=\frac{\alpha_1-\frac{1}{2}}{\alpha_1-\alpha_n},$$
where the respective probabilities are computed as in Section~\ref{sec:preliminaries}.

We now consider three cases. First, if $\frac{p_{\fa}}{p_{\fb}}=\frac{p_1}{p_n}$, then let $\eta=p_1+p_n$, and let $s_2$ be the signal whose realizations are $\{r_2,\ldots,r_{n-1}\}$, and where $\p{s_2(\theta)=r_i} = \p{s(\theta)=r_i}/(1-\eta)$ for each $i\in\{2,\ldots,n-1\}$ and $\theta\in\{\theta_A,\theta_B\}$. Observe that $s_2$ yields posterior beliefs $(\alpha_2,\ldots,\alpha_{n-1})$. It is straightforward to verify that these satisfy the claim of the lemma.

Second, suppose $\frac{p_{\fa}}{p_{\fb}}<\frac{p_1}{p_n}$. In this case, let $\eta =  \frac{p_{\fa}}{p_{\fb}}p_n+p_n$, and let $s_2$ be the signal whose realizations are $\{r_1,r_2,\ldots,r_{n-1}\}$,  where $\p{s_2(\theta)=r_i} = \p{s(\theta)=r_i}/(1-\eta)$ for each for each $i\in\{2,\ldots,n-1\}$ and $\theta\in\{\theta_A,\theta_B\}$, and where $\p{s_2(\theta)=r_1} = \left(\p{s(\theta)=r_1}-\eta\p{s_1(\theta)=r_1}\right)/(1-\eta)$ for each  $\theta\in\{\theta_A,\theta_B\}$. Observe that $s_2$ yields posterior beliefs $(\alpha_1,\alpha_2,\ldots,\alpha_{n-1})$. Again, it is straightforward to verify that these satisfy the claim of the lemma. 
In particular, the only shared realization $r_1$ induces the same posterior belief under both signals: $\p{\theta_A|s_1=r_1}=\p{\theta_A|s_2=r_1}$.

Finally, suppose $\frac{p_{\fa}}{p_{\fb}}>\frac{p_1}{p_n}$. This case is analogous to the second case, except that $\eta = p_1 +  \frac{p_{\fb}}{p_{\fa}}p_1$,
 $r_n$ replaces $r_1$ in $\supp(s_2)$, the signal $s_2$ yields beliefs $(\alpha_2,\ldots,\alpha_{n-1}, \alpha_n)$, and $s_1$ and $s_2$ share the realization $r_n$.
\end{proof}

\bibliography{socialMediaDemocracy}

\newcommand{\noop}[1]{}
\begin{thebibliography}{39}
\expandafter\ifx\csname natexlab\endcsname\relax\def\natexlab#1{#1}\fi
\expandafter\ifx\csname url\endcsname\relax
  \def\url#1{\texttt{#1}}\fi
\expandafter\ifx\csname urlprefix\endcsname\relax\def\urlprefix{URL }\fi
\providecommand{\eprint}[2][]{\url{#2}}

\bibitem[{Acemoglu et~al.(2011)Acemoglu, Dahleh, Lobel and
  Ozdaglar}]{acemoglu2011bayesian}
\textsc{Acemoglu, D.}, \textsc{Dahleh, M.~A.}, \textsc{Lobel, I.} and
  \textsc{Ozdaglar, A.} (2011).
\newblock Bayesian learning in social networks.
\newblock \textit{The Review of Economic Studies}, \textbf{78} 1201--1236.

\bibitem[{Alonso and C{\^a}mara(2016)}]{alonso2016persuading}
\textsc{Alonso, R.} and \textsc{C{\^a}mara, O.} (2016).
\newblock Persuading voters.
\newblock \textit{American Economic Review}, \textbf{106} 3590--3605.

\bibitem[{Arieli and Babichenko(2019)}]{arieli2019private}
\textsc{Arieli, I.} and \textsc{Babichenko, Y.} (2019).
\newblock Private bayesian persuasion.
\newblock \textit{Journal of Economic Theory}, \textbf{182} 185--217.

\bibitem[{Arieli et~al.(2020)Arieli, Babichenko and
  Smorodinsky}]{arieli2020identifiable}
\textsc{Arieli, I.}, \textsc{Babichenko, Y.} and \textsc{Smorodinsky, R.}
  (2020).
\newblock Identifiable information structures.
\newblock \textit{Games and Economic Behavior}, \textbf{120} 16--27.

\bibitem[{Arieli et~al.(2021)Arieli, Gradwohl and Smorodinsky}]{arieli2021herd}
\textsc{Arieli, I.}, \textsc{Gradwohl, R.} and \textsc{Smorodinsky, R.} (2021).
\newblock Herd design.
\newblock \textit{Available at SSRN 3917729}.

\bibitem[{Austen-Smith and Banks(1996)}]{austen1996information}
\textsc{Austen-Smith, D.} and \textsc{Banks, J.~S.} (1996).
\newblock Information aggregation, rationality, and the condorcet jury theorem.
\newblock \textit{American Political Science Review}, \textbf{90} 34--45.

\bibitem[{Chan et~al.(2019)Chan, Gupta, Li and Wang}]{chan2019pivotal}
\textsc{Chan, J.}, \textsc{Gupta, S.}, \textsc{Li, F.} and \textsc{Wang, Y.}
  (2019).
\newblock Pivotal persuasion.
\newblock \textit{Journal of Economic theory}, \textbf{180} 178--202.

\bibitem[{De~Condorcet(1785)}]{de2014essai}
\textsc{De~Condorcet, N.} (1785).
\newblock \textit{Essai sur l'Application de l'Analyse {\`a} la Probabilit{\'e}
  des D{\'e}cisions Rendues {\`a} la Pluralit{\'e} des Voix}.
\newblock Paris. Reprinted by Cambridge University Press, 2014.

\bibitem[{Denter et~al.(2021)Denter, Dumav and Ginzburg}]{denter2021social}
\textsc{Denter, P.}, \textsc{Dumav, M.} and \textsc{Ginzburg, B.} (2021).
\newblock Social connectivity, media bias, and correlation neglect.
\newblock \textit{The Economic Journal}, \textbf{131} 2033--2057.

\bibitem[{Downs(1957)}]{downs1957economic}
\textsc{Downs, A.} (1957).
\newblock \textit{An Economic Theory of Democracy}.
\newblock Harper \& Row New York, NY.

\bibitem[{Facebook(2022)}]{facebook}
\textsc{Facebook} (2022).
\newblock {Facebook News Feed: An introduction for content creators}.
\newblock
  \url{https://www.facebook.com/business/learn/lessons/facebook-news-feed-creators}.
\newblock [Online; accessed 12-June-2022].

\bibitem[{Feddersen and Pesendorfer(1997)}]{feddersen1997voting}
\textsc{Feddersen, T.} and \textsc{Pesendorfer, W.} (1997).
\newblock Voting behavior and information aggregation in elections with private
  information.
\newblock \textit{Econometrica: Journal of the Econometric Society} 1029--1058.

\bibitem[{Feddersen and Pesendorfer(1998)}]{feddersen1998convicting}
\textsc{Feddersen, T.} and \textsc{Pesendorfer, W.} (1998).
\newblock Convicting the innocent: The inferiority of unanimous jury verdicts
  under strategic voting.
\newblock \textit{American Political science review}, \textbf{92} 23--35.

\bibitem[{Gentzkow et~al.(2015)Gentzkow, Shapiro and Stone}]{gentzkow2015media}
\textsc{Gentzkow, M.}, \textsc{Shapiro, J.~M.} and \textsc{Stone, D.~F.}
  (2015).
\newblock Media bias in the marketplace: Theory.
\newblock In \textit{Handbook of Media Economics}, vol.~1. Elsevier, 623--645.

\bibitem[{Gitmez and Molavi(2022)}]{gitmez2022polarization}
\textsc{Gitmez, A.~A.} and \textsc{Molavi, P.} (2022).
\newblock Polarization and media bias.
\newblock \textit{arXiv preprint arXiv:2203.12698}.

\bibitem[{Grossman and Helpman(1999)}]{grossman1999competing}
\textsc{Grossman, G.~M.} and \textsc{Helpman, E.} (1999).
\newblock Competing for endorsements.
\newblock \textit{American Economic Review}, \textbf{89} 501--524.

\bibitem[{Grossman and Helpman(2001)}]{grossman2001special}
\textsc{Grossman, G.~M.} and \textsc{Helpman, E.} (2001).
\newblock \textit{Special Interest Politics}.
\newblock MIT press.

\bibitem[{Heese and Lauermann(2021)}]{heese2021persuasion}
\textsc{Heese, C.} and \textsc{Lauermann, S.} (2021).
\newblock Persuasion and information aggregation in elections.
\newblock Tech. rep., ECONtribute Discussion Paper.

\bibitem[{Hillman(2010)}]{hillman2010expressive}
\textsc{Hillman, A.~L.} (2010).
\newblock Expressive behavior in economics and politics.
\newblock \textit{European Journal of political economy}, \textbf{26} 403--418.

\bibitem[{Hillman and Ursprung(1988)}]{hillman1988domestic}
\textsc{Hillman, A.~L.} and \textsc{Ursprung, H.~W.} (1988).
\newblock Domestic politics, foreign interests, and international trade policy.
\newblock \textit{American Economic Review} 729--745.

\bibitem[{Hotelling(1929)}]{hotelling1929stability}
\textsc{Hotelling, H.} (1929).
\newblock Stability in competition.
\newblock \textit{The Economic Journal}, \textbf{39} 41--57.

\bibitem[{Kamenica(2019)}]{kamenica2019bayesian}
\textsc{Kamenica, E.} (2019).
\newblock Bayesian persuasion and information design.
\newblock \textit{Annual Review of Economics}, \textbf{11} 249--272.

\bibitem[{Kamenica and Gentzkow(2011)}]{kamenica2011bayesian}
\textsc{Kamenica, E.} and \textsc{Gentzkow, M.} (2011).
\newblock Bayesian persuasion.
\newblock \textit{American Economic Review}, \textbf{101} 2590--2615.

\bibitem[{Kerman et~al.(2020)Kerman, Herings and Karos}]{kerman2020persuading}
\textsc{Kerman, T.}, \textsc{Herings, P. J.-J.} and \textsc{Karos, D.} (2020).
\newblock Persuading strategic voters.

\bibitem[{Koriyama and Szentes(2009)}]{koriyama2009resurrection}
\textsc{Koriyama, Y.} and \textsc{Szentes, B.} (2009).
\newblock A resurrection of the condorcet jury theorem.
\newblock \textit{Theoretical Economics}, \textbf{4} 227--252.

\bibitem[{Kuran(1997)}]{kuran1997private}
\textsc{Kuran, T.} (1997).
\newblock \textit{Private Truths, Public Lies: The Social Consequences of
  Preference Falsification}.
\newblock Harvard University Press.

\bibitem[{Levy(2021)}]{levy2021social}
\textsc{Levy, R.} (2021).
\newblock Social media, news consumption, and polarization: Evidence from a
  field experiment.
\newblock \textit{American Economic Review}, \textbf{111} 831--70.

\bibitem[{Lichter(2017)}]{lichter2017theories}
\textsc{Lichter, S.~R.} (2017).
\newblock Theories of media bias.
\newblock In \textit{The Oxford Handbook of Political Communication}. Oxford
  University Press New York, NY, 403--416.

\bibitem[{Mukhopadhaya(2003)}]{mukhopadhaya2003jury}
\textsc{Mukhopadhaya, K.} (2003).
\newblock Jury size and the free rider problem.
\newblock \textit{Journal of Law, Economics, and Organization}, \textbf{19}
  24--44.

\bibitem[{Pons and Tricaud(2018)}]{pons2018expressive}
\textsc{Pons, V.} and \textsc{Tricaud, C.} (2018).
\newblock Expressive voting and its cost: Evidence from runoffs with two or
  three candidates.
\newblock \textit{Econometrica}, \textbf{86} 1621--1649.

\bibitem[{Prat(2002)}]{prat2002campaign}
\textsc{Prat, A.} (2002).
\newblock Campaign spending with office-seeking politicians, rational voters,
  and multiple lobbies.
\newblock \textit{Journal of Economic Theory}, \textbf{103} 162--189.

\bibitem[{Puglisi and Snyder~Jr(2015)}]{puglisi2015empirical}
\textsc{Puglisi, R.} and \textsc{Snyder~Jr, J.~M.} (2015).
\newblock Empirical studies of media bias.
\newblock In \textit{Handbook of media economics}, vol.~1. Elsevier, 647--667.

\bibitem[{Spenkuch(2018)}]{spenkuch2018expressive}
\textsc{Spenkuch, J.~L.} (2018).
\newblock Expressive vs. strategic voters: An empirical assessment.
\newblock \textit{Journal of Public Economics}, \textbf{165} 73--81.

\bibitem[{Str{\"o}mberg et~al.(2015)}]{stromberg2015media}
\textsc{Str{\"o}mberg, D.} \textsc{et~al.} (2015).
\newblock Media and politics.
\newblock \textit{Annual Review of Economics}, \textbf{7} 173--205.

\bibitem[{Sun et~al.(2019)Sun, Schram and Sloof}]{sun2019theory}
\textsc{Sun, J.}, \textsc{Schram, A.} and \textsc{Sloof, R.} (2019).
\newblock A theory on media bias and elections.
\newblock Tech. rep., Tinbergen Institute Discussion Paper.

\bibitem[{Taneva(2019)}]{taneva2019information}
\textsc{Taneva, I.} (2019).
\newblock Information design.
\newblock \textit{American Economic Journal: Microeconomics}, \textbf{11}
  151--85.

\bibitem[{Twitter(2022{\natexlab{a}})}]{twitter2}
\textsc{Twitter} (2022{\natexlab{a}}).
\newblock {Help with locked or limited account}.
\newblock
  \url{https://help.twitter.com/en/managing-your-account/locked-and-limited-accounts}.
\newblock [Online; accessed 12-June-2022].

\bibitem[{Twitter(2022{\natexlab{b}})}]{twitter1}
\textsc{Twitter} (2022{\natexlab{b}}).
\newblock {Notices on Twitter and what they mean}.
\newblock
  \url{https://help.twitter.com/en/rules-and-policies/notices-on-twitter}.
\newblock [Online; accessed 12-June-2022].

\bibitem[{Wang(2013)}]{wang2013bayesian}
\textsc{Wang, Y.} (2013).
\newblock Bayesian persuasion with multiple receivers.
\newblock \textit{Available at SSRN 2625399}.

\end{thebibliography}

\end{document}